\documentclass[journal,draftclsnofoot,onecolumn,12pt]{IEEEtran}
\usepackage{amsthm,amssymb,graphicx,multirow,amsmath,color,amsfonts}
\usepackage[update,prepend]{epstopdf}
\usepackage[noadjust]{cite}
\usepackage[latin1]{inputenc}
\usepackage{tikz}
\usepackage{bbm} 
\usepackage{pdfpages}
\usepackage{multirow}
\usepackage{subfig}
\usepackage{comment}
\def\nb0{{\mathbf{0}}}
\def\nb1{{\mathbf{1}}}

\newtheorem{lemma}{Lemma}

\newtheorem{definition}{Definition}

\newtheorem{theorem}{Theorem}

\newtheorem{remark}{Remark}

\newtheorem{approximation}{Approximation}	

\usepackage{booktabs}
\usepackage[nocomma]{optidef}
\usepackage{algorithm} 
\usepackage{algpseudocode}

\begin{document}
\title{
A Game-Theoretic Framework for Coexistence of WiFi and Cellular Networks in the 6-GHz Unlicensed Spectrum
}


\author{
Aniq Ur Rahman, Mustafa A. Kishk, {\em Member, IEEE} and\\ Mohamed-Slim Alouini, {\em Fellow, IEEE}
\thanks{Aniq Ur Rahman, Mustafa A. Kishk and Mohamed-Slim Alouini are with King Abdullah University of Science and Technology (KAUST), CEMSE division, Thuwal 23955-6900, Kingdom of Saudi Arabia.} \thanks{E-mail: \{aniqur.rahman, mustafa.kishk, slim.alouini\}@kaust.edu.sa} 
\vspace{-4mm}}

\maketitle

\begin{abstract}
    We study the interaction of WiFi and 5G cellular networks as they exploit the recently unlocked 6 GHz spectrum for unlicensed access while conforming to the constraints imposed by the incumbent users. We use tools from stochastic geometry to derive the theoretical performance metrics for users of each radio access technology, which helps us capture the aggregate behaviour of the network in a snapshot. We propose a framework where the portions of cellular and WiFi networks are grouped to form entities. These entities interact to satisfy their Quality of Service demands by playing a non-cooperative game. The action of an entity corresponds to the fraction of its network elements (WiFi access point and cellular base stations) operating in the 6 GHz band. Due to the decentralized nature of the game, we find the solution using distributed Best Response Algorithm (\texttt{D-BRA}), which improves the average datarate by $11.37\%$ and $18.59\%$ for cellular and WiFi networks, respectively, with random strategy as the baseline. The results demonstrate how the system parameters affect the performance of a network at equilibrium and highlight the throughput gains of the networks as a result of using the 6 GHz bands, which offer considerably larger bandwidths. We tested our framework on a real-world setup with the actual locations for both the cellular network and WiFi access point, which shows that practical implementation of multi-entity spectrum sharing is feasible even when the spatial distribution of the network elements and population are non-homogeneous.
\end{abstract}



\section{Introduction}

\IEEEPARstart{F}{ederal} Communications Commission (FCC) has recently unlocked the 6~GHz band (from 5.925~GHz to 7.125~GHz) for unlicensed users in US~\cite{fcc2020new, fcc2020new2}.  {This \textit{beachfront} 6-GHz spectrum has also been unlocked for unlicensed use in United Kingdom, Saudi Arabia, South Korea, Canada, Brazil and many countries in the European Union~\cite{ofcom2021, broadcom2021, warburton_2021}.} This unlicensed spectrum will be used by both WiFi users (Wi-Fi 6E) and cellular users (5G new radio unlicensed or NR-U) {\cite{3gpp_ref12, 3gpp_ref14, sathya2020standardization}}. In addition, restriction will be applied on unlicensed users to ensure no interference is experienced by the incumbent users~\cite{fwcc2019}, which includes fixed point communications such as wireless backhaul~\cite{wif2019}. 
In order to study this new system setup, two types of coexistence needs to be taken into consideration: (i) the coexistence between the unlicensed users and the incumbent users, and (ii) the coexistence between WiFi and 5G users. For the former, as stated earlier, the restrictions deployed by standardization entities will govern the interplay between licensed and unlicensed users. However, such restrictions do not exist for the latter. In particular, WiFi and 5G users, while respecting the aforementioned restrictions, will both try to adjust their system parameters to maximize their Quality of Service (QoS), such as coverage or throughput.

Conventionally, the users prefer WiFi over cellular connection to access faster and more reliable internet~\cite{jerry2020park}, but this behaviour is changing, as the cost of cellular internet is decreasing, thereby incentivizing its use over WiFi~\cite{oughton2020revisiting}.
Moreover cellular base stations cover a larger area compared to WiFi access points, which means that handovers are not that frequent, even for mobile users like vehicles and aerial drones. These differences in the use cases hint at the coexistence of these technologies instead of one taking over the other. 
The convergence between WiFi and 5G will open new business avenues and ultimately improve the user experience \cite{wbangmnalliance2019}, as the two radio access technologies (RATs) will simultaneously complement each other while competing for the available bandwidth \cite{oughton2020revisiting}. Therefore, improvement in the performance of both RATs is essential for transitioning into the future smoothly.

{The proponents of the \textit{Digital Inclusion} movement see spectrum sharing and unlicensed spectrum usage as potential solutions to make communication technologies affordable for rural and low-income neighborhoods \cite{chaoub20206g, montsi2017real}.
The application of Listen-Before-Talk (LBT) schemes which work well for cognitive radios is not justified when the multiple networks are allowed  to exploit the unlicensed spectrum without any preferential treatment. This calls for an intelligent distributed algorithm which allocates the spectrum efficiently and minimizes the interference~\cite{chaoub20206g}.
}

In this work, the network is divided among multiple entities~\cite{sacoto2020game}, where each entity encompasses a portion of the WiFi and a portion of the cellular network. An entity can decide what fraction of its network elements, namely, the WiFi access points (APs) and cellular base stations (BSs), will operate in the unlicensed 6 GHz band. The variation in this fraction affects the performance of the entity's networks. Therefore, the choice of this fraction must be optimal. Moreover, the decision of one entity affects the performance of other entities. This implies that the entities will take a series of decisions, one after the other, and finally settle at a decision which is mutually optimal for all the entities. The interaction of these entities can be modelled as a non-cooperative multi-player game. The performance of the networks of an entity and their respective QoS requirements are used to define the payoff function of that entity.

\subsection{Related Works}

\subsubsection{Coexistence in the Unlicensed Spectra}
Recent surveys \cite{jerry2020park, sathya2020standardization}, have reviewed the coexistence of WiFi and cellular networks in the 6~GHz unlicensed bands, and provided an extensive list of relevant literature.
Fair coexistence of LTE with WiFi has been studied for the 5~GHz band \cite{naik2018coexistence}. In \cite{cano2016unlicensed}, the authors define fairness in terms of datarate. Proportional fairness demonstrated in \cite{mehrnoush2018fairness}, allows each RAT to access the unlicensed channel for equal amount of time, which is better in contrast to the notion of fairness championed by 3GPP.
Recently, coexistence of 5G NR with WiFi is also being investigated in the 60~GHz mmWave band \cite{lagen2019new, wang2020unlicensed}.
{Bao et al.~\cite{bao2021joint} realize the coexistence of WiFi and cellular networks in unlicensed bands by time-sharing and power-allocation.}

\subsubsection{Stochastic Geometry for Studying Coexistence}
Authors in \cite{pinto2009stochastic} make use of tools from stochastic geometry to study the coexistence of narrow band (NB) and ultra-wide band (UWB) wireless nodes.
In \cite{bhorkar2014performance}, the authors analyse the performance of an LTE-WiFi network coexisting in the unlicensed 5~GHz band using stochastic geometry.
{Authors in \cite{naikcoexistence} study the coexistence of WiFi and NR-U in the 6 GHz unlicensed bands and focus on uplink transmissions. The work makes use of stochastic geometry and reports performance superior to CSMA/CA. }

\subsubsection{Spectrum Sharing using Game Theory}
Spectrum sharing is of key significance in alleviating mutual interference and it allows multiple users to utilize the spectrum in parallel \cite{voicu2016}.
The cognitive radio literature \cite{wu2008repeated, wang2008game, teng2017sharing, hasan2016lte, cognitive2014} treats the open spectrum sharing problem as a game, where the secondary users compete for the unlicensed spectrum. The users' actions include varying certain parameters such as transmission power, access duration and modulation technique. The payoff is generally modelled as a function of the experienced quality of service, such as throughput or latency. Bairagi et al.~\cite{bairagi2018qoe} adopt a game-theoretic framework to formulate the problem of unlicensed spectrum sharing among WiFi and LTE networks. First, coexistence is achieved through time-sharing between LTE and WiFi users by the Kalai-Smorodinsky bargaining solution. Next, the resources are allocated to LTE users by a Q-learning based algorithm.
{Authors in \cite{chen2014social} formulate the problem of white space spectrum sharing among users as a social group utility maximization (SGUM) game where each user tries to maximize the utility of its social group.}

\subsection{Contributions}
The major contributions of our work are as follows:
\begin{itemize}
    \item The behaviour of WiFi and 5G cellular networks is studied, as they utilize the unlicensed 6~GHz spectrum in addition to their licensed/legacy spectra, while respecting the constraints imposed by the incumbent users. The interaction of multiple entities is formulated as a non-cooperative game and the payoff function of each entity is defined as a weighted sum of the datarates achieved by the RAT(s) it owns. The payoff is non-zero only when the datarate requirements of the RAT(s) are satisfied. 
    
    \item The coverage probability and average datarate for each radio access technology's users is derived using tools from stochastic geometry. This helps in capturing the aggregate behaviour of the network in a snapshot and in expressing the performance metrics theoretically based on the actions of the entities.
    
    \item The repeated normal-form game is solved through an iterative distributed algorithm, called the distributed best response algorithm (\texttt{D-BRA}), in which the entities take action without any time synchronization or centralized scheduling.
    
    \item The framework developed\footnote{{The simulation and analysis codes are available at \cite{coexgithub} (to be made public post-acceptance).}} in this study is highly flexible and can be used to analyze a variety of scenarios, across various unlicensed bands.
    
    \item {We tested the proposed framework on real-world network of Glasgow, UK where the spatial distribution of population and the network elements are non-homogeneous. The results suggest that a practical implementation of our framework is feasible to improve the coexistence of WiFi and cellular networks utilizing the 6-GHz unlicensed spectrum.}
\end{itemize}


The rest of the paper is organized as follows. In Sec.~II we describe the system model and then develop a mathematical understanding of the performance metrics in the next section, Sec.~III. In Sec.~IV we propose a game theoretic framework to study the problem, and show the results in Sec.~V. {In Sec.~VI, we test our framework on a real-world network and assess the performance.} The paper is finally concluded in Sec.~VII.




%
\section{System Model}
\label{sec:system_model}
We consider a system where three coexisting wireless networks are utilizing the 6-GHz band: (i) the WiFi network consisting of the APs and the WiFi users, (ii) the cellular network consisting of cellular BSs and its users, and (iii) a network of the incumbent users, utilizing the 6-GHz band licensed to them for fixed point backhaul operation. We consider this to be the primary usage of this band and the performance of the first two networks should not hamper the experience of the incumbent users. To enforce this, each of these incumbent users have an \textit{exclusion zone} around them, within which the 6-GHz band cannot be used by the WiFi and cellular networks~\cite{fcc2020new, fcc2020new2}.

The WiFi APs and Cellular BSs are divided among various \textit{entities}, i.e., parts of the network are owned by different mobile operators. The entities control their own utilization of the unlicensed spectrum to maximize their datarate in response to the 6-GHz spectrum utilization by the other entities. The interaction between the entities is modelled as a non-cooperative game. The system is illustrated in Fig.~\ref{fig:sysmod} and a sample network deployment scenario is presented in Fig.~\ref{fig:netdep}.

\begin{figure}[h!]
    \centering
    \includegraphics[width=0.6\columnwidth]{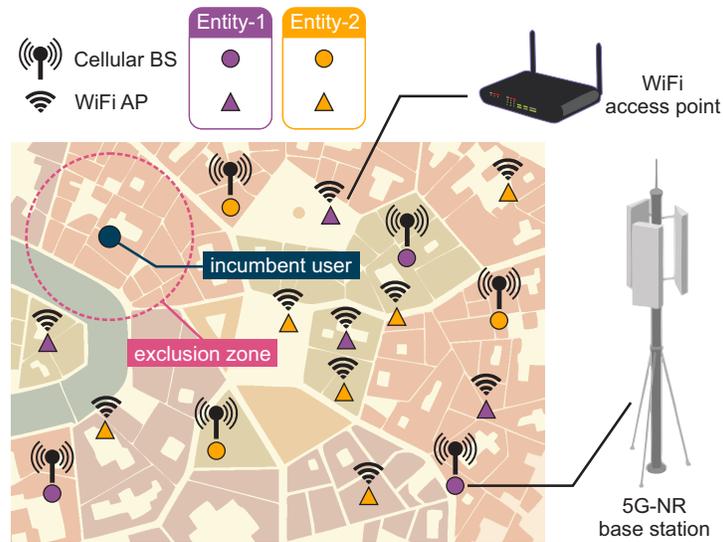}
    \caption{System Illustration.}
    \label{fig:sysmod}
\end{figure}

\subsection{Network Deployment}
\subsubsection{Incumbent Users}
We populate the $\mathbb{R}^2$ plane with incumbent users, which follow a homogeneous Poisson point process (PPP) $\Phi_z$ with parameter $\lambda_z$. Around each of these incumbent users, we have an exclusion zone of radius $\rho$. The set of all the exclusion zones is described mathematically as: $\Xi_{\rho} = \bigcup_{ \mathbf{x} \in \Phi_z} b(\mathbf{x}, \rho),$ where $b(\mathbf{x}, \rho)$ denotes a disk of radius $\rho$ centred at $\mathbf{x}$.

\subsubsection{WiFi Network}
Now we deploy the WiFi access points (APs) as a homogeneous PPP $\Phi_w$ of intensity $\lambda_w$. The set of APs which are allowed to use the unlicensed 6~GHz spectrum lie outside the exclusion zones $\Xi_{\rho}$ and can be carved out as a Poisson Hole Process (PHP) from $\Phi_w$. We define this PHP as $\hat{\Phi}_w \triangleq \Phi_w \setminus \Xi_{\rho}$.

\begin{approximation}
The PHP $\hat{\Phi}_w$ can be approximated as a uniform PPP  of intensity $\bar{\lambda}_w = \lambda_w \exp(- \pi \lambda_z  \rho^2)$, see \cite{kishk2017tight, kishk2017coexistence}.
\end{approximation}

A fraction $\delta_w \in [0,1]$ of APs in $\hat{\Phi}_w$ operate in the unlicensed band. Therefore, we can write: $\Phi_w = \Phi_{w|L} \bigcup \Phi_{w|U}$, where $\Phi_{w|L}$ is the set of APs operating in their legacy 2.4GHz WiFi band, and $\Phi_{w|U}$ is set of APs using the unlicensed 6-GHz band. Furthermore, we can approximate $\Phi_{w|L}$ and $\Phi_{w|U}$ as independent homogeneous PPPs, such that: $\Phi_{w|U} \triangleq {\rm PPP}(\delta_w \bar{\lambda}_w)$, and $\Phi_{w|L} \triangleq {\rm PPP}( \lambda_w - \delta_w \bar{\lambda}_w)$, where ${\rm PPP}(\Lambda)$ denotes a homogeneous Poisson point process having intensity $\Lambda$.

The WiFi users form a cluster point process, where each cluster has a radius $\rho_w$, and is centred around an AP. The WiFi users are denoted as $\Psi_w$ and described as: $\Psi_w = \bigcup_{\mathbf{x}\in \Phi_w} \psi_{\mathbf{x}} + \mathbf{x}$, where $\psi_{\mathbf{x}}$ is a Mat\'ern cluster centred around $\mathbf{x}$. In a Mat\'ern cluster, the distance between a point in the cluster to the center follows a triangular distribution bounded within $(0, \rho_w)$, see \cite{hayajneh2018performance}.

\begin{figure}[t!]
    \centering
    \includegraphics[width=0.5\columnwidth]{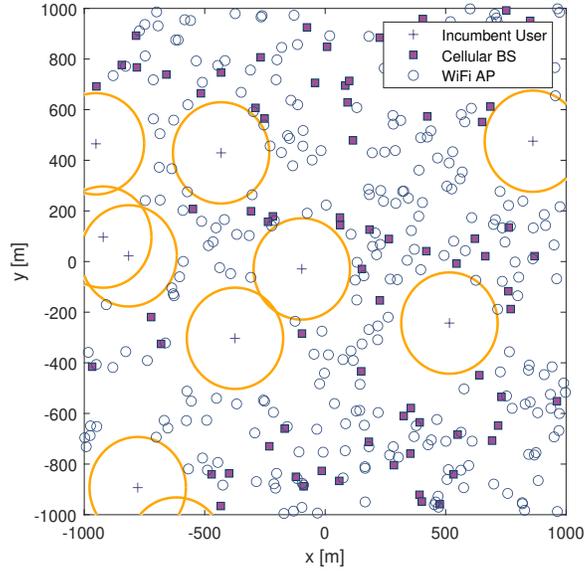}
    \caption{Sample network deployment showing $\Phi_z, \Xi_{\rho}, \hat{\Phi}_c, \hat{\Phi}_w$. The exclusion zones of radius $\rho$ are drawn around the incumbent users as yellow circles.  Parameters: $\lambda_z= 1~{\rm user}/{\rm km}^2$, $\lambda_c= 25~{\rm BS}/{\rm km}^2$, $\lambda_w=100~{\rm AP}/{\rm km}^2$, $\rho = 200~{\rm m}$.}
    \label{fig:netdep}
\end{figure}

\subsubsection{Cellular Network}
Similar to the deployment of the WiFi APs, we model the set of cellular Base stations as an independent homogeneous PPP $\Phi_{c}$ of intensity $\lambda_c$, such that only the BSs lying outside the exclusion zones are permitted to utilize the unlicensed 6-GHz spectrum.
The set of BSs outside the exclusion zones is a PHP $\hat{\Phi}_c = \Phi_{c} \setminus \Xi_{\rho} $ and a fraction $\delta_c \in [0,1]$ of BSs in $\hat{\Phi}_c$ operate in the unlicensed band. 

\begin{approximation}
The PHP $\hat{\Phi}_c$ is approximated as a homogeneous PPP with intensity $\bar{\lambda}_c = \lambda_c \exp(- \pi \lambda_z  \rho^2)$, see \cite{kishk2017tight, kishk2017coexistence}.
\end{approximation}

Alternatively, $\Phi_{c}$ can be described as: $\Phi_{c} = \Phi_{c|L} \bigcup \Phi_{c|U}$, where $\Phi_{c|L} \triangleq {\rm PPP}( \lambda_c - \delta_c \bar{\lambda}_c)$ and $\Phi_{c|U} \triangleq {\rm PPP}(\delta_c \bar{\lambda}_c)$ operate in the licensed and unlicensed bands, respectively. 
Furthermore, the cellular users are spread in $\mathbb{R}^2$ as an independent homogeneous PPP.

The transmit powers of the cellular BSs, WiFi APs and incumbent users are $p_c$, $p_w$ and $p_z$ respectively. The bandwidth offered by the unlicensed 6-GHz band is $B_U$. The bandwidths of the cellular and WiFi networks operating in their corresponding licensed/legacy bands is denoted as $B_{c|L}$ and $B_{w|L}$, respectively. We also denote the bandwidths of the cellular and WiFi users in the unlicensed band as $B_{c|U}$ and $B_{w|U}$ respectively, and both are equal to $B_U$ in value.

\subsection{Downlink Interference}
\begin{definition}
The user at the origin receives a signal of strength $\xi(\mathbf{x})$ from the transmitting node at $\mathbf{x}$: $\xi(\mathbf{x})= p_{\mathbf{x}}H \|\mathbf{x} \|^{-\alpha}$, where $p_{\mathbf{x}}$ is the transmit power of the node $\mathbf{x}$ and $H \sim \exp(1)$ is the random variable signifying the channel gain due to Rayleigh fading and $\alpha$ is the path-loss coefficient.
\end{definition}

In \textit{unlicensed access}, the typical user experiences interference from all the networks utilizing the unlicensed 6-GHz spectrum, namely, $\Phi_z$, $\Phi_{c|U}$ and $\Phi_{w|U}$.
We define this set as $\Phi^U \triangleq \Phi_z \cup \Phi_{c|U} \cup \Phi_{w|U}.$ The interference to cellular and WiFi users in unlicensed access is denoted as $I_{c|U}$ and $I_{w|U}$, respectively, and defined as follows:
\begin{align}
    I_{c|U} &\triangleq \sum_{\mathbf{x} \in \Phi^U \setminus \{ \mathbf{x}_0 \}} \xi(\mathbf{x}), &\mathbf{x}_0 = \arg \min_{ \mathbf{x}\in \Phi_{c|U}} \| \mathbf{x} \| ,\\
    I_{w|U} &\triangleq \sum_{\mathbf{x} \in \Phi^U \setminus \{ \mathbf{x}_0 \}} \xi(\mathbf{x}), &\mathbf{x}_0 \sim \{ \mathbf{x} \in \Phi_{w|U} : \| \mathbf{x} \| < \rho_w \}. \label{IwU}
\end{align}

In \textit{licensed or legacy access}, the typical user experiences interference from its own network. The cellular and WiFi users receive interference from the nodes in $\Phi_{c|L}$ and $\Phi_{w|L}$ respectively. The interference to cellular and WiFi users in licensed/legacy access is denoted as $I_{c|L}$ and $I_{w|L}$ respectively, and defined as follows:
\begin{align}
    I_{c|L} &\triangleq \sum_{\mathbf{x} \in \Phi_{c|L} \setminus \{ \mathbf{x}_0 \}} \xi(\mathbf{x}), &\mathbf{x}_0 = \arg \min_{ \mathbf{x}\in \Phi_{c|L}} \| \mathbf{x} \|,\\
    I_{w|L} &\triangleq \sum_{\mathbf{x} \in \Phi_{w|L} \setminus \{ \mathbf{x}_0 \}} \xi(\mathbf{x}), &\mathbf{x}_0 \sim \{ \mathbf{x} \in \Phi_{w|L} : \| \mathbf{x} \| < \rho_w \}.
    \label{IwL}
\end{align}

It must be noted that unlike the cellular users, the WiFi users do not connect to the nearest access point. Instead, they are connected to any single AP while being within its range. This is expressed mathematically in equations \eqref{IwU} and  \eqref{IwL} as $\mathbf{x}_0 \sim \{ \mathbf{x} \in \Phi_{w|M} : \| \mathbf{x} \| < \rho_w \}$, $M \in \{ U, L \},$ where $\mathbf{x}_0$ is the access point, to which the WiFi user connects. Next, we define the signal-to-interference-plus-noise ratio (SINR), which is ultimately used for the downlink analysis in the upcoming section.

\begin{definition}
The signal-to-interference-plus-noise ratio for a typical user is defined as: ${\rm SINR}_{k|M} \triangleq \dfrac{\xi(\mathbf{x}_0)}{ \sigma^2_k + I_{k|M}}$;  $k\in\{c,w\}$, $M\in\{U,L\}$, where $\sigma_c^2$ and $\sigma_w^2$ are the receiver noise power for the cellular and WiFi users respectively.
\end{definition}

\subsection{Multi-Entity Competition}
{In simple words, an entity can be thought of as a Mobile Network Operator (MNO) which uses multiple WiFi access points (APs) and cellular base stations (BSs) to provide service to its subscribers.}
It is seen in \cite{sacoto2020game} that Mobile Virtual Network Operators (MVNOs) which are a conglomeration of multiple MNOs can generate more revenue compared to multiple MNOs operating independently. Therefore, we can envisage {the entity being} an MVNO having both RATs under its umbrella, which then allocates the resources judiciously in order to achieve the best overall performance. {For example, Google has launched its own MVNO by the name of Google Fi~\cite{google2021} which aims at improving the QoS by smartly switching between WiFi and cellular networks.}

We begin by defining a set of entities, $\texttt{e}_i \in \mathcal{E}$ which consists of a cellular network and a WiFi network. Entity $\texttt{e}_i$'s share of the cellular network is $v_c^i$ and its share of WiFi network is $v_w^i$, such that, $\sum_{\texttt{e}_j \in \mathcal{E}} v_c^j = 1$, and $\sum_{\texttt{e}_j \in \mathcal{E}} v_w^j = 1$.
The network of type $k \in \{c,w\}$ owned by entity $\texttt{e}_i$ is denoted as $\Phi_k^i\triangleq{\rm PPP}(v_k^i\,\lambda_k)$. The portion of $\Phi_k^i$ which lies outside the exclusion zones is denoted as $\hat{\Phi}_k^i\triangleq{\rm PPP}(v_k^i \, \bar{\lambda}_k)$. 
A fraction $\delta_c^i$ of the cellular base stations $\in \hat{\Phi}_c^i$ and a fraction  $\delta_w^i$ of the WiFi APs $\in \hat{\Phi}_w^i$ operate in the unlicensed 6-GHz band. The fraction of network elements in $\Phi_k^i$ which operate in the licensed/legacy and unlicensed bands are denoted as $\Phi_{k|L}^i\triangleq{\rm PPP}(\delta_k^i\,v_k^i\,\bar{\lambda}_k)$, and $\Phi_{k|U}^i\triangleq {\rm PPP}(  v_k^i \, \lambda_k - \delta_k^i\, v_k^i \,\bar{\lambda}_k)$, respectively.
Moreover, each entity has its own QoS requirement in the form of minimum datarates for its WiFi and cellular networks. The minimum datarate requirements of $\texttt{e}_i$ for the cellular and WiFi networks are denoted as $\hat{\sigma}_c^i$ and $\hat{\sigma}_w^i$,  respectively. The concept of entities is illustrated in Fig.~\ref{fig:entity}. Finally, we present a formal definition for an entity as follows:
\begin{definition}
An entity $\texttt{e}_i \in \mathcal{E}$ can be defined by the tuple $\left( v_c^i, v_w^i, \delta_c^i, \delta_w^i, \hat{\sigma}_c^i,\hat{\sigma}_w^i \right)$, with the following description:
    (1) $v_c^i$ is the cellular network share,
    (2) $v_w^i$ is the WiFi network share,
    (3) $\delta_c^i$ is the cellular unlicensed spectrum utilization,
    (4) $\delta_w^i$ is the WiFi unlicensed spectrum utilization,
    (5) $\hat{\sigma}_c^i$ is the  cellular datarate threshold,
    (6) $\hat{\sigma}_w^i$ is the  WiFi datarate threshold.

\end{definition}

\begin{figure}[h!]
    \centering
    \includegraphics[width=0.6\columnwidth]{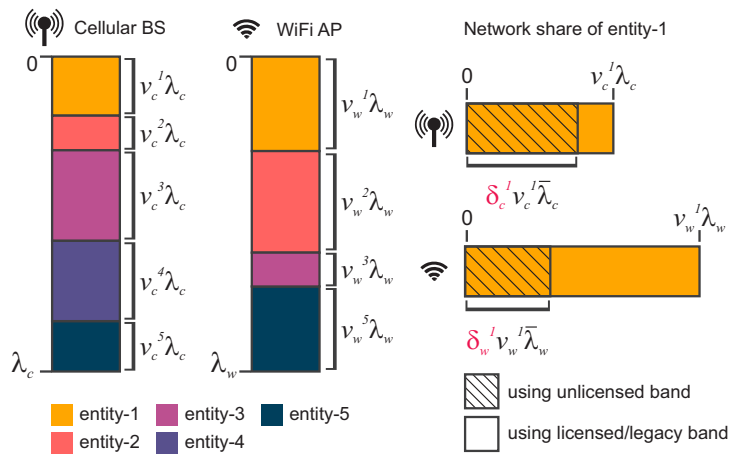}
    \caption{Division of the network among a set of Entities.}
    \label{fig:entity}
\end{figure}

In this model, $v_c^i =0$ implies that the entity lacks a cellular network and $v_w^i =0$ denotes the absence of a WiFi network. 
In the context of the entire network, the fraction denoting the utilization of the unlicensed band by the cellular and WiFi networks is defined as $\delta_c \triangleq \sum_{\texttt{e}_j \in \mathcal{E}} v_c^j \delta_c^j$ and $\delta_w \triangleq \sum_{\texttt{e}_j \in \mathcal{E}} v_w^j \delta_w^j$ respectively.


\section{Performance Metrics}
In this section we derive the performance metrics, which are used to define the payoff function in the game-theoretic formulations in the next section. 

\subsection{Coverage Probability}
We begin by presenting the theoretical expressions of coverage probability for the cellular and WiFi networks operating in the different bands based on the parameters discussed in Sec.~\ref{sec:system_model}. Formally, we define the coverage probability as follows:
\begin{definition}
Coverage Probability is defined as the probability that the SINR experienced by a user is greater than the threshold value of $\gamma$, i.e., $P_{k|M} \triangleq \mathbb{P}\left( {\rm SINR}_{k|M} > \gamma \right),$ where $k \in \{c,w\}$ denotes the network (cellular/WiFi), and $M \in \{U,L\}$ denotes the mode of access (U for unlicensed and L for licensed/legacy).
\end{definition}

\subsubsection{Cellular Users}
The coverage probability for cellular users operating in the licensed and unlicensed bands are described in Lemma~\ref{lemma1} and Lemma~\ref{lemma2} respectively, as follows.

\begin{lemma}
The coverage probability for a cellular user operating in the licensed band is:
\begin{align*}
    P_{c|L}(\gamma, \delta_c) &= 
    2 \pi \left( \lambda_c - \delta_c \bar{\lambda}_c \right) \int_{0}^{\infty} \exp\Bigg\{ - \frac{\kappa_c^2 \gamma}{p_c} r^{\alpha} 
    - \pi \left( \lambda_c - \delta_c \bar{\lambda}_c \right) \left(1 + \zeta(\gamma, \alpha) \right) r^2 \Bigg\} \, r \, dr,
\end{align*}
where $\zeta(\gamma, \alpha) =  \frac{\gamma^{\frac{2}{\alpha}}}{2} \int_{\gamma^{-\frac{2}{\alpha}}}^{\infty} \frac{1}{1 + x^{\frac{\alpha}{2}}} \, dx.$ 
\label{lemma1}
\end{lemma}

\begin{proof}
The proof is provided in \cite[Proposition 5.2.3]{haenggi2018stochastic} and is hence skipped.
\end{proof}

\begin{lemma}
The coverage probability for a cellular user in the 6-GHz unlicensed band is:
\begin{align*}
    P_{c|U}(\gamma, \delta_c, \delta_w) &=  2 \pi \delta_c  \bar{\lambda}_c \int_{0}^{\infty} \exp\Bigg\{ - \frac{\kappa_c^2 \gamma}{p_c} r^{\alpha} -\Bigg( 
    \frac{\pi \gamma^{2/\alpha}}{p_c^{2/\alpha} {\rm sinc}\left(\frac{2}{\alpha}\right)} \Big( \delta_w\bar{\lambda}_w p_w^{\frac{2}{\alpha}}
    + \lambda_z p_z^{\frac{2}{\alpha}}  \Big)\nonumber\\ 
    &  + \pi \delta_c  \bar{\lambda}_c \left(1 + \zeta(\gamma, \alpha) \right)
    \Bigg)r^2 \Bigg\} \, r \, dr; \quad \text{where} \quad \zeta(\gamma, \alpha) =  \frac{\gamma^{\frac{2}{\alpha}}}{2} \int_{\gamma^{-\frac{2}{\alpha}}}^{\infty} \frac{1}{1 + x^{\frac{\alpha}{2}}} \, dx.
\end{align*} 
\label{lemma2}
\end{lemma}
 
\begin{proof}
Using the definition of SINR and coverage probability above, we express the coverage probability for a cellular user operating in the unlicensed band as: $P_{c|U}(\gamma, \delta_c, \delta_w) = \mathbb{P}\left(\frac{p_c H R^{-\alpha}}{\kappa^2_c + I_{c|U}} > \gamma \right)$.
Then we take the expectation over the two random quantities: the distance between the cellular user and its base station $R$, and the interference $I_{c|U}$. The distance $R$ follows the distribution:
    $f_R^c(r) = 2\pi \delta_c \bar{\lambda}_c r \, \exp(-\pi \delta_c \bar{\lambda}_c r^2), \quad r \geq 0.$
Further, exploiting the CDF of $H$ which is exponentially distributed as $H\sim \exp(1)$, we arrive at:
\begin{align*}
    P_{c|U} = \int_{0}^{\infty} e^{-\kappa_c^2p_c^{-1}\gamma r^{\alpha}} \times  \prod_{j \in \{ w,z,c \}} \mathcal{L}_{c,j|U}\left(\frac{r^{\alpha}\gamma}{p_c} \right) \times f_R^c(r) \, dr,
\end{align*}
where $\mathcal{L}_{c,j|U}(s)$ is the Laplace transform of the interference experienced by the cellular users in the unlicensed band from transmitters of type $j \in \{c,w,z\}$. Plugging the expression for the Laplace transforms from the Appendix yields the final equation in the lemma.
\end{proof}

\subsubsection{WiFi Users}
The coverage probability for WiFi users operating in the legacy and unlicensed bands are described in Lemma~\ref{lemma3} and Lemma~\ref{lemma4} respectively, as follows.

\begin{lemma}
The coverage probability for a WiFi user operating in the legacy WiFi band is:
\begin{align*}
    P_{w|L}(\gamma, \delta_w) = \frac{2}{\rho_w^2}  \int_{0}^{\rho_w} \exp\Bigg\{ - \frac{\kappa_w^2 \gamma}{p_w} r^{\alpha} -\frac{\pi \gamma^{2/\alpha} \left(\lambda_w - \delta_w\bar{\lambda}_w \right) }{{\rm sinc}\left(\frac{2}{\alpha}\right)} r^2 \Bigg\} \, r \, dr.
\end{align*}
\label{lemma3}
\end{lemma}
 
\begin{proof}
The proof is provided in \cite[Proposition 5.2.1]{haenggi2018stochastic} and is hence skipped.
\end{proof}

\begin{lemma}
The coverage probability for a WiFi user operating in the 6-GHz unlicensed band is: 
\begin{align*}
    P_{w|U}(\gamma, \delta_c, \delta_w) = \frac{2}{\rho_w^2}  \int_{0}^{\rho_w} \exp\Bigg\{ - \frac{\kappa_w^2 \gamma}{p_w} r^{\alpha} -\frac{\pi \gamma^{2/\alpha}}{p_w^{2/\alpha} {\rm sinc}\left(\frac{2}{\alpha}\right)} \Big(\delta_w \bar{\lambda}_w p_w^{\frac{2}{\alpha}} + \delta_c \bar{\lambda}_c p_c^{\frac{2}{\alpha}} 
    + \lambda_z p_z^{\frac{2}{\alpha}}  \Big)r^2 \Bigg\} \, r \, dr.
\end{align*}
\label{lemma4}
\end{lemma}
 
\begin{proof}
The proof can be sketched on the same lines as shown in the proof of Lemma~\ref{lemma2}.
Here, the distance $R$ between a WiFi user and the access point it is associated with, follows the distribution: $f_R^w(r) = \frac{2r}{\rho_w^2} \cdot \mathrm{1} \{ 0 \leq r<\rho_w \}$.
\end{proof}

In Fig.~\ref{fig:cov_match}, we see that our theoretical expressions for coverage probability described in Lemmas 1-4, are in agreement with the Monte Carlo simulation results.

\begin{figure}[t!]
    \centering
    \includegraphics[width=0.6\columnwidth]{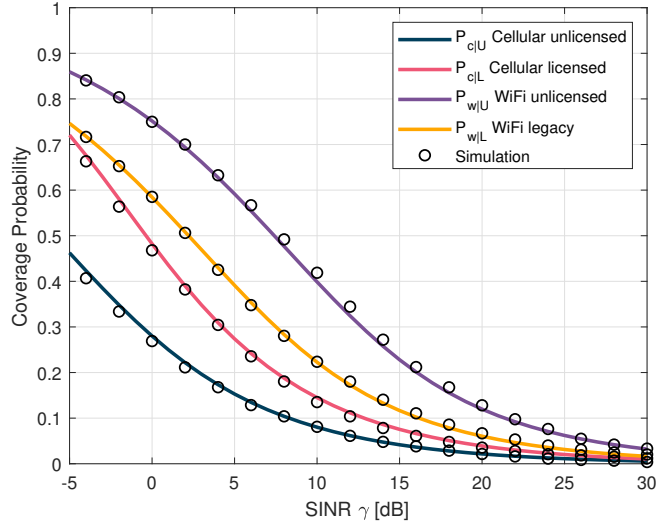}
    \caption{Coverage probability for the cellular and WiFi users operating in the licensed and unlicensed spectra. $\lambda_z= 1~{\rm user}/{\rm km}^2$, $\lambda_c= 25~{\rm BS}/{\rm km}^2$, $\lambda_w=100~{\rm AP}/{\rm km}^2$, $\rho = 200~{\rm m}$, $\rho_w = 50~{\rm m}$, $p_z = 1~{\rm W}$, $p_c = 2~{\rm W}$, $p_w = 1~{\rm W}$, $\delta_c = 0.7$, $\delta_w = 0.2$.}
    \label{fig:cov_match}
\end{figure}

\subsection{Average Datarate}
In this section, we define the average datarate of typical cellular and WiFi users.
Formally, it is defined as the average of the datarates experienced by the users associated with the network elements of an entity, either cellular or WiFi.
In the following two theorems, we present theoretical expressions for the average datarate of the cellular and WiFi users, derived using Stochastic Geometry.
 
\begin{theorem}
The average datarate of a user served by the cellular network of entity $\texttt{e}_i$ is:
\begin{align}
    &\sigma_c^i (\gamma, \delta_c, \delta_w, \delta_c^i) =
    B_{c|U} \log_2 \left( 1 + \gamma \right) P_{c|U}(\gamma, \delta_c, \delta_w)\cdot \frac{\delta_c^i  \bar{\lambda}_c}{\lambda_c} \nonumber \\ 
    &\qquad + B_{c|L} \log_2 \left( 1 + \gamma \right) P_{c|L}(\gamma, \delta_c)\cdot \left( 1 - \frac{\delta_c^i  \bar{\lambda}_c}{\lambda_c} \right).
\end{align}
\label{theorem1}
\end{theorem} 

\begin{proof}
We can express the time-averaged datarate of a cellular user associated with a base station $\mathbf{x} \in \Phi_{c}^i$ of entity $\texttt{e}_i$ as: $B_{\mathbf{x}} \log_2 (1 + \gamma) \, \mathbb{P} ( {\rm SINR}_{\mathbf{x}} > \gamma )$, where $B_{\mathbf{x}}$ is the bandwidth offered by $\mathbf{x}$. Taking expectation over $\Phi_c^i$, we get:
\begin{align*}
    \sigma_c^i &\stackrel{(a)}{=} \sum_{M \in \{U, L\}} \mathbb{E}_{\Phi_{c|M}^i } \Big[ B_{\mathbf{x}} \log_2 (1 + \gamma) \, \mathbb{P} ( {\rm SINR}_{\mathbf{x}} > \gamma )  \Big] \\
    &\stackrel{(b)}{=}  \sum_{M \in \{U, L\}} \frac{|\Phi_{c|M}^i|}{|\Phi_{c}^i|} \, B_{c|M} \log_2(1 + \gamma) \, P_{c|M}
\end{align*}
Step~(a) follows from $\Phi_c^i = \Phi_{c|U}^i \cup \Phi_{c|L}^i$. In step~(b), $|\Phi|$ denotes the intensity of a PPP $\Phi$.
\end{proof}

\begin{theorem}
The average datarate of a user served by the WiFi network of entity $\texttt{e}_i$ is:
\begin{align}
    &\sigma_w^i (\gamma, \delta_c, \delta_w, \delta_w^i) = B_{w|U} \log_2 \left( 1 + \gamma \right) P_{w|U}(\gamma, \delta_c, \delta_w)\cdot \frac{\delta_w^i  \bar{\lambda}_w}{\lambda_w} \nonumber \\
    &\qquad + B_{w|L} \log_2 \left( 1 + \gamma \right) P_{w|L}(\gamma, \delta_w)\cdot \left( 1 - \frac{\delta_w^i  \bar{\lambda}_w}{\lambda_w}\right).
\end{align}
\label{theorem2}
\end{theorem}

\begin{proof}
The proof can be sketched on the same lines as Theorem~\ref{theorem1}, and is hence skipped.
\end{proof}

In Fig.~\ref{fig:avg_dat}, we present the average cellular and WiFi datarates for various values of $\delta_c$ and $\delta_w$. The cellular datarate peaks at $\delta_c = 1$ and $\delta_w = 0$ implying that it performs the best in the absence of interference from the WiFi network in the unlicensed band. Moreover, for any value of $\delta_w$, the maximum cellular datarate is observed at $\delta_c = 1$.
\begin{figure}[h!]
    \centering
    \includegraphics[width=0.6\columnwidth]{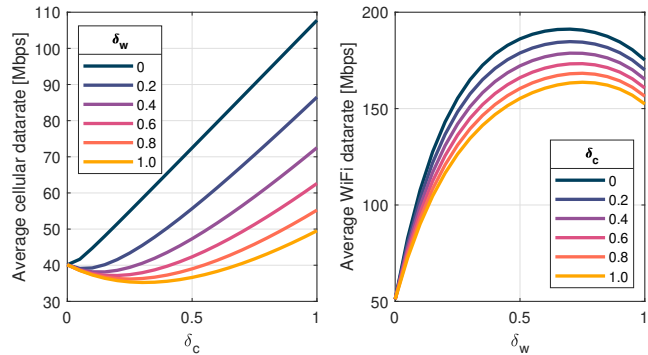}
    \caption{Average cellular and WiFi datarates. Parameters: $|\mathcal{E}|=1$, $\lambda_z= 1~{\rm user}/{\rm km}^2$, $\lambda_c= 25~{\rm BS}/{\rm km}^2$, $\lambda_w=100~{\rm AP}/{\rm km}^2$, $\rho = 200~{\rm m}$, $\rho_w = 50~{\rm m}$, $p_z = 1~{\rm W}$, $p_c = 2~{\rm W}$, $p_w = 1~{\rm W}$, $B_U= 240~{\rm MHz}$, $B_{c|L}= 80~{\rm MHz}$, $B_{w|L}= 80~{\rm MHz}$, $\gamma = 10~{\rm dB}$.}
    \label{fig:avg_dat}
\end{figure}
However, the WiFi datarate has a maximum at $\delta_w \approx 0.7$ for any value of $\delta_c$ indicating that it performs best when the self-interference in the unlicensed band by the WiFi network is kept under check. In other words, if all the WiFi APs operate in the unlicensed band, the interference increases and aggravates the datarate. This suggests that the WiFi networks should under-utilize the 6-GHz spectrum for best results.
{These trends indicate a non-linear relationship between WiFi and cellular datarates based on the variation in 6-GHz spectrum utilization ($\delta_c, \delta_w$) by the entities. Therefore a dynamic policy is required which can settle at the  optimal $\delta_c, \delta_w$ values as the environment evolves.}

In the next section, we invoke the concept of non-cooperative games to model the interaction between different entities, where each entity tries to maximize the value of its payoff function by taking the necessary action.


\section{Game Formulation}
Each entity $\texttt{e}_i$ adjusts its unlicensed spectrum utilization $\delta_c^i, \delta_w^i$ to maximize its own payoff function. The interaction of all the entities can be modelled as a non-cooperative game.
Therefore we represent the action of entity $\texttt{e}_i$ as a vector $\mathbf{a}_i$ of unlicensed spectrum utilization factors as:
$\mathbf{a}_i \triangleq [\delta_c^i \quad \delta_w^i ]^T \in \mathcal{C}_{\mu}^2$,
{where $\mathcal{C}_{\mu}= \{0, \mu, 2\mu, \cdots, 1\}$ is the discretized version of $[0,1]$ with a step-size of $\mu$ to make the action space finite, thereby making the game finite.}
The set of actions of all the entities except $\texttt{e}_i$ is denoted as:
$\mathcal{A}_{-i} \triangleq \left\{ \mathbf{a}_j: \forall \texttt{e}_j \in \mathcal{E} \setminus \{ \texttt{e}_i \}   \right\}$.

\begin{definition}
The payoff function of entity $\texttt{e}_i$ is defined as $f_i(\mathbf{a}_i, \mathcal{A}_{-i} ) :  {\mathcal{C}_{\mu}}^{2|\mathcal{E}|} \mapsto \mathbb{R}$, which maps the actions of all the entities to a real value.
\end{definition}

\begin{definition}
The game is defined as $\mathfrak{G}\left( \mathcal{E}, \mathcal{A}, \mathcal{F} \right)$, where entities $\mathcal{E}$ are the players, $\mathcal{A} \triangleq {\mathcal{C}_{\mu}}^{2|\mathcal{E}|}$ is the action profile and $\mathcal{F} \triangleq \left( f_i , \cdots f_{|\mathcal{E}|} \right)$ is the set of payoff functions of all players.
\end{definition}

In this non-cooperative framework, we impose a strict condition that the entities are not aware of the actions of other entities. When an entity takes an action, it is based on the observations from the environment and the action then alters the environment.
Each network element of an entity can measure the power of the interfering signals in the licensed and unlicensed bands, and leverage this information to estimate the average datarate experienced by the users in its proximity.

Since there is no coordination among the entities, we cannot schedule the decision making process as a round-robin sequence. Instead we equip each entity with an independent Poisson clock \cite{durand2018distributed} with the same rate, which triggers the entity to take an action (see Fig.~\ref{fig:poisson_trigger}). We also assume that the time taken by each entity to choose an action is small enough to {make the probability of multiple entities acting simultaneously negligible.}

\begin{figure}[h!]
    \centering
    \includegraphics[width=0.6\columnwidth]{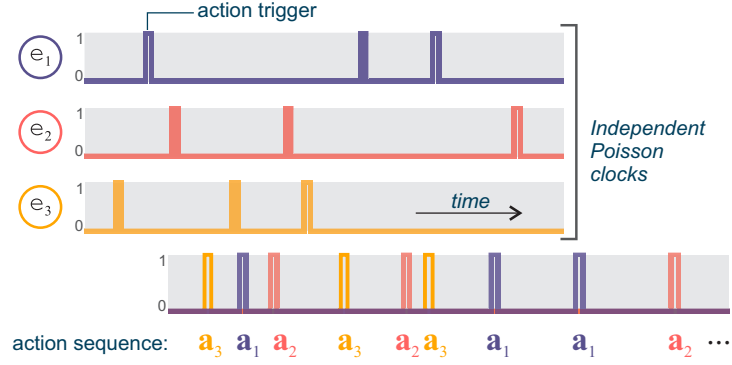}
    \caption{Action sequence as a result of independent Poisson clocks.}
    \label{fig:poisson_trigger}
\end{figure}

\begin{remark}
The action sequence generated as a consequence of each entity being triggered by its Poisson clock is a sequence where each element is the action of an entity drawn from the set of all entities with uniform probability.
\label{clock}
\end{remark}

\subsection{Payoff Function}
Next, we mathematically define the payoff function which the entities can adopt to find their best response.
The condition that the cellular and WiFi datarates are greater than the minimum values, can be indicated as $1_{\sigma_i} \triangleq \mathrm{1} \left\{ 
    \bigcap_{k \in \{c,w\}} \left( \sigma_k^i \geq \hat{\sigma}_k^i | v_k^i > 0 \right) \right\}$.
This ensures that the action which fulfils the minimum datarate will always be chosen in the presence of actions which might have a higher weighted sum of datarates but do not meet the datarate requirements for both networks simultaneously.

The \textit{preference} for the cellular and WiFi networks in an entity $\texttt{e}_i$ is translated by the weights $\theta_c^i > 0$ and $\theta_w^i > 0$ respectively, which are used in the definition of the payoff function below.
\begin{align}
    f_i \triangleq  1_{\sigma_i} \times  \sum_{k \in  \{c,w\}} \theta_k^i \sigma_k^i \, \mathrm{1}\{ v_k^i > 0 \}.
    \label{maxdat}
\end{align}

{The concept of network preference terms $\theta_c$ and $\theta_w$ is introduced to balance the datarates of the WiFi and cellular networks and promote fair coexistence. The value of these weights can be set by the respective entity based on their own analysis. The effect of $\theta_c$ and $\theta_w$ is further explained in the results section (Sec.~V-A).}


\newpage
\subsection{Distributed Algorithm}

{We focus on finding the set of actions $[ \mathbf{a_1} \cdots \mathbf{a_k} \cdots \mathbf{a_{|\mathcal{E}|}} ]$ which the entities collectively do not change in their subsequent iterations.}

\begin{definition}
The solution of $\mathfrak{G}$ is defined as the matrix of action vectors of all the entities $\begin{bmatrix} \mathbf{a}_1^{[n_1]} \cdots \mathbf{a}_k^{[n_k]} \cdots \mathbf{a}_{|\mathcal{E}|}^{[n_{|\mathcal{E}|}]}  \end{bmatrix}$ such that the action vector of each entity remains unchanged in their latest iterations, i.e., $\sum\limits_{\texttt{e}_j \in \mathcal{E}} \left\lVert \mathbf{a}_j^{[n_j]} - \mathbf{a}_j^{[n_j-1]} \right\rVert  = 0$, where $\mathbf{a}_j^{[m]}$ denotes the action vector of $\texttt{e}_j$ at the $m^{\rm th}$ iteration and $n_j$ is the latest iteration of entity $\texttt{e}_j$.
\end{definition}

{The solution is equivalent to the  Nash equilibrium, as the only reason why the players stop changing their action vectors collectively is because they have reached an `equilibrium'. Since the payoff functions are not straightforward, we avoid finding the solution analytically and therefore resort to the distributed best response algorithm} (\texttt{D-BRA}) which is presented in Algorithm~\ref{DBRA}. The algorithm can be set to run at periodic intervals. This makes the system adaptive to environmental changes due to the addition or removal of nodes, or the introduction of rogue elements which {perturb} the equilibrium from time to time.

\begin{algorithm}[h!]
\caption{Distributed Best Response Algorithm (\texttt{D-BRA})}
\begin{algorithmic}[1]
    \State $\mathbf{a}_j \sim {\mathcal{C}_{\mu}^2}; \quad \forall \texttt{e}_j \in \mathcal{E}$ \Comment{\small{Initialization}}
    
    \State $n_j \gets 0; \quad \forall \texttt{e}_j \in \mathcal{E}$
    
    \While{ $\sum\limits_{\texttt{e}_j \in \mathcal{E}} \left\lVert \mathbf{a}_j^{[n_j]} - \mathbf{a}_j^{[n_j-1]} \right\rVert  > 0 $ } 
    
    \State $i \sim {\rm uniform}\left( 1, |\mathcal{E}| \right)$ \Comment{\small{Remark~\ref{clock}}}
    
    \State $ \mathbf{a}_i \gets \arg \max_{\mathbf{a} \in {\mathcal{C}_{\mu}^2}} f_i ( \mathbf{a}, \mathcal{A}_{-i}) $ \Comment{\small{Best Response: {solved through exhaustive search in $\mathcal{C}_{\mu}^2$}} }
    
    \State $n_i \gets n_i + 1$
    
    \State $\mathbf{A}_i^{[n_i]} \gets \mathbf{a}_i$ \Comment{\small{Update action vector}}

    \EndWhile
\end{algorithmic}
\label{DBRA}
\end{algorithm}

\begin{remark}
{When the game $\mathfrak{G}$ does not have a \textit{pure strategy} Nash equilibrium, \texttt{D-BRA} does not terminate. In such cases, if we analyze the action vectors of all the agents over a large number of iterations, the probability mass function of the action space $\mathcal{A}$ converges in distribution\footnote{supported by \cite[Proposition 2.2]{fudenberg1998theory}.} to the \textit{mixed strategy} Nash equilibrium, which is guaranteed to exist by the Folk theorem \cite{demeze2020complete}, since  $\mathfrak{G}$ can be viewed as a repeated game in normal-form.}
\end{remark}

{In practice, we can have $\sum\limits_{\texttt{e}_j \in \mathcal{E}} \left\lVert \mathbf{a}_j^{[n_j]} - \mathbf{a}_j^{[n_j-1]} \right\rVert  \leq \varepsilon$, where $\varepsilon > 0$ is an arbitrarily small value which controls the accuracy of the solution.} {This way, we get an approximate solution in the $\varepsilon$-neighborhood of the exact solution which can make the algorithm converge faster, albeit to a near-optimal solution. To improve the accuracy of the solution further, we can reduce the value of the step-size $\mu$ to refine the search space.}
{Intuitively, the complexity of \texttt{D-BRA} increases with increase in number of entities $|\mathcal{E}|$ and search space size $|\mathcal{C}_{\mu}|$ but the exact characterization deserves a thorough mathematical investigation and is therefore beyond the scope of this work.}


\section{Results \& Discussion}
We present the results for the interaction between the entities and analyse how the various parameters affect the datarate achieved by the entities at equilibrium, in order to provide useful design insights for implementation. The range of the parameters used in the system model are summarized in Table~\ref{tab:symbols}. 
For all the results that follow, we use the following values: $\lambda_z= 1~{\rm user}/{\rm km}^2$, $\lambda_c= 25~{\rm BS}/{\rm km}^2$, $\lambda_w=100~{\rm AP}/{\rm km}^2$, $\rho = 200~{\rm m}$, $\rho_w = 50~{\rm m}$, $p_z = 1~{\rm W}$, $p_c = 2~{\rm W}$, $p_w = 1~{\rm W}$, $B_U= 240~{\rm MHz}$, $B_{c|L}= 80~{\rm MHz}$, $B_{w|L}= 80~{\rm MHz}$, $\gamma = 10~{\rm dB}$, $\mu = 0.1$.

\begin{table}[!ht]
\caption{Range of System Parameters}
\label{tab:symbols}
\centering
\begin{tabular}{ c c c }
\toprule
\textbf{Parameter} & \textbf{Values} & \textbf{Source}\\ \midrule
$B_U$ & $\in$ [40, 80, 160, 240, 320] MHz & \cite{sathya2020standardization, jerry2020park} \\ 
$B_{c|L}$ & $\in$ [20, 40, 80, 100] MHz & \cite{3gpp_38_101} \\
$B_{w|L}$ & $\in$ [20, 40, 80, 160] MHz & \cite{lopez2019ieee} \\\hline
$p_z$ & $\leq$ 30 dBm & \cite{sathya2020standardization}\\ 
$p_c$ & $\leq$ 36 dBm & \cite{sathya2020standardization}\\
$p_w$ & $\leq$ 36 dBm & \cite{sathya2020standardization}\\ \hline
$\lambda_z$ & 1 ${\rm km}^{-2}$ & \cite{radiocommunicationbureau}\\ 
$\lambda_c$ & $\in$ [25, 50, 250] ${\rm km}^{-2}$ & \cite{alanweissberger2020, ding2018fundamental} \\
$\lambda_w$ &  $\in$ [100, 400] ${\rm km}^{-2}$ &  \cite{jones2007wi} \\ \hline 
$\rho$ & 200 m & \\
$\rho_w$ &  50 m & \cite{sun2020discrete} \\ 
\bottomrule
\end{tabular}
\end{table}

\subsection{Effect of Network Preference}
We define two entities, each demanding a minimum cellular datarate of 30~Mbps and minimum WiFi datarate of 100~Mbps. 
To visualize the behaviour of the payoff function, we vary the ratio $\theta_c^i/\theta_w^i$ and plot the values of cellular and WiFi datarates summarized for $(v_c^i, v_w^i) \in [0.1, 0.9]^2$. In Fig.~\ref{fig:ratio_dat}, the mean value of the average cellular datarate increases with an increase in the ratio. This implies that as the entity gives more preference to its cellular network, its performance improves. An interesting trend is observed for the WiFi datarates. After a dip at $\theta_c = 2 \theta_w$, it rises back up $\theta_c = 5 \theta_w$ onwards. {This can be explained by looking at how the 6-GHz spectrum is utilized for $\theta_c/\theta_w=1$ case. Here the cellular networks rely mostly on their licensed band, which reduces interference to the WiFi networks operating in the 6 GHz unlicensed bands, thereby increasing their datarate, as observed.}
\begin{figure}[h!]
    \centering
    \includegraphics[width=0.9\columnwidth]{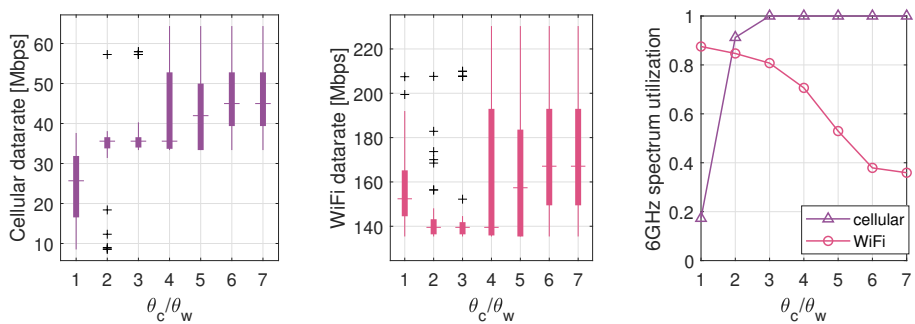}
    \caption{Average cellular and WiFi datarates at Equilibrium for different values of ${\theta_c}/{\theta_w}$. $|\mathcal{E}|=2$, $\hat{\sigma}_c = 30$~Mbps, $\hat{\sigma}_w = 100$~Mbps, $\gamma=10$~dB.}
    \label{fig:ratio_dat}
\end{figure}

In Fig.~\ref{fig:2ent_tune_ratio}, for the same entities playing the same game, the Nash equilibrium changes due to difference in the value of the ratio $\frac{\theta_c}{\theta_w}$.
In Fig.~\ref{tune:a}, we give 7 times more preference to the cellular network datarate compared to its WiFi counterpart, due to which the cellular datarate of both entities settle at $~\sim 55$~Mbps which satisfies the datarate requirement of both the entities. Moreover, the equilibrium is satisfactory for both networks of each entity. 
\begin{figure}[h!]%
 \centering
 \subfloat[$\theta_c = 7 \theta_w$ ]{\includegraphics[trim={2.5cm 0 2.5cm 0},clip,width=0.7\columnwidth]{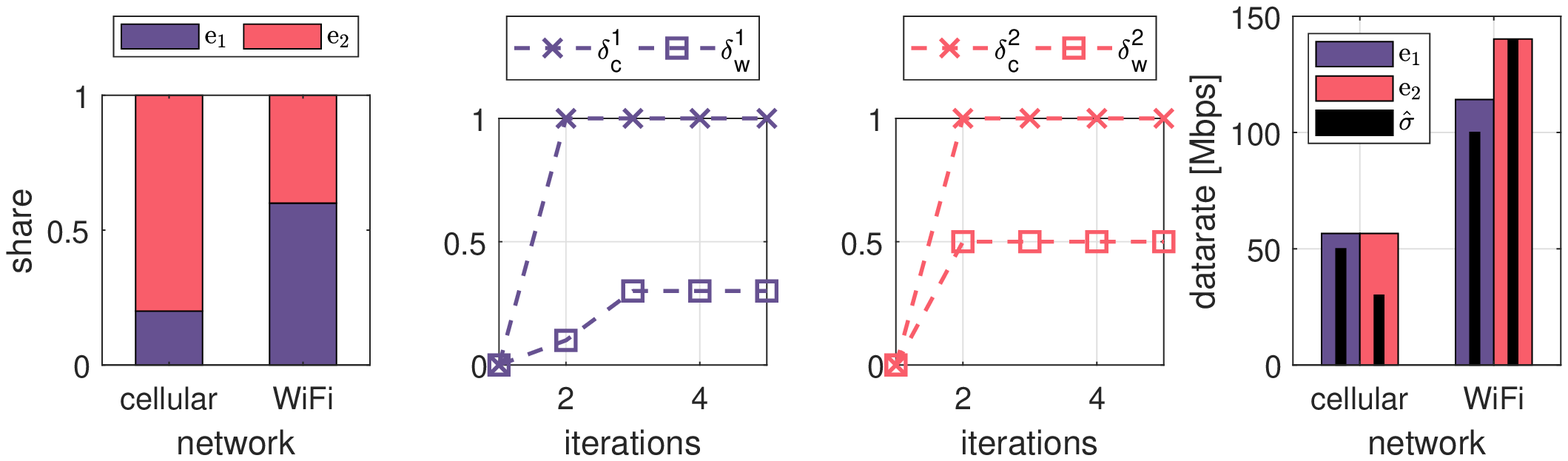}\label{tune:a}}\\
 \subfloat[$\theta_c = 5 \theta_w$ ]{\includegraphics[trim={2.5cm 0 2.5cm 0},clip,width=0.7\columnwidth]{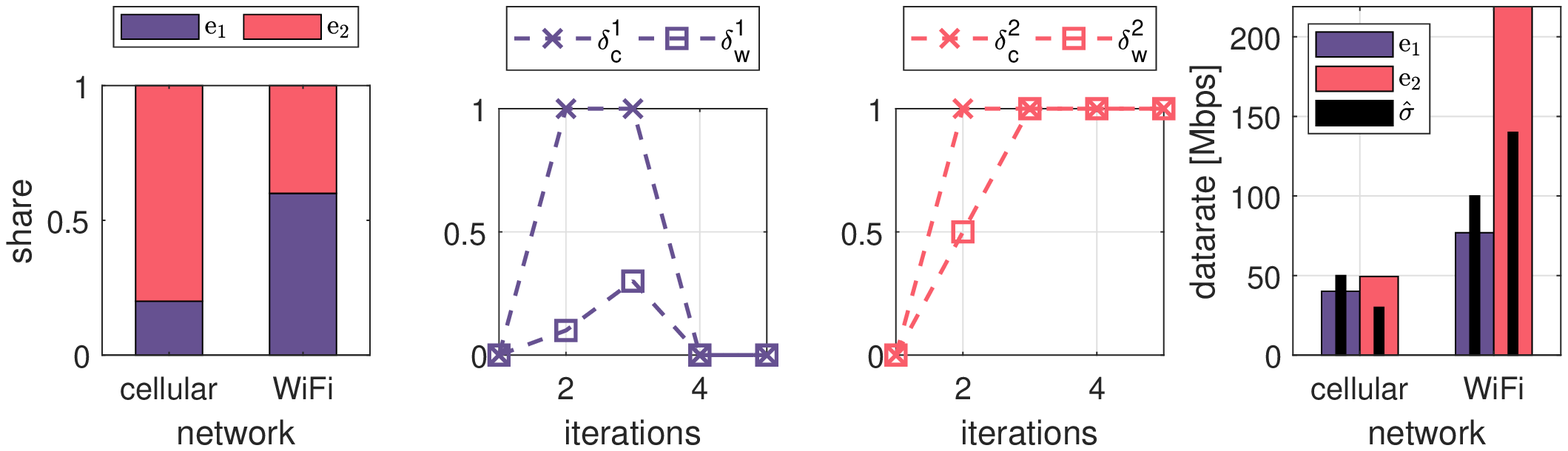}\label{tune:b}}
 \caption{Nash Equilibria due to difference in $\theta_c/\theta_w$. $|\mathcal{E}|=2$.}%
 \label{fig:2ent_tune_ratio}%
\end{figure}
In contrast, in Fig.~\ref{tune:b} where $\theta_c = 5\theta_w$ we observe that the datarate requirements of entity-1 are not met at all. This could be attributed to entity-2 greedily setting $\delta_w^2$ value to 1 to maximize its WiFi datarate. In the subsequent iterations, none of the actions by entity-1 could satisfy its datarates, so it saturates at $\delta_c^1 =0, \delta_w^1 = 0$. Therefore tuning the value of the ratio $\frac{\theta_c}{\theta_w}$ is essential for system performance.


\subsection{Effect of Datarate Thresholds}
We study two entities with identical datarate requirements and vary $(v_c^1, v_w^1) \in [0.1, 0.9]^2$. Then, we analyse the empirical \textit{complementary cumulative density function} (CCDF) of the cellular and WiFi datarates for different values of minimum cellular and WiFi datarate thresholds. The value of this CCDF at $x$, is referred to as the \textit{rate coverage probability} (RCP) at datarate threshold equal to $x$.

We first comment on the variation in average cellular datarate due to different datarate thresholds as shown in Fig.~\ref{fig:cell_thresh} and we observe the following:
\begin{itemize}
    \item The datarate requirement is perfectly met when the WiFi datarate threshold $\hat{\sigma}_w$ is low ($100$~Mbps) as well as the cellular datarate threshold $\hat{\sigma}_c$ is low ($30$~Mbps). 
    \item Setting the cellular datarate threshold to a value higher than the true requirement ensures a higher rate coverage probability (RCP) at the true requirement. For example, when $\hat{\sigma}_c = 50$~Mbps and $\hat{\sigma}_w = 180$~Mbps, the RCP at $50$~Mbps is $\sim 0.2$, but it shoots up to $\sim 0.9$ at $30$~Mbps.
    \item The RCP decreases with increase in the WiFi datarate threshold $\hat{\sigma}_w$.
\end{itemize}

\begin{figure}[h!]
    \centering
    \includegraphics[width=0.6\columnwidth]{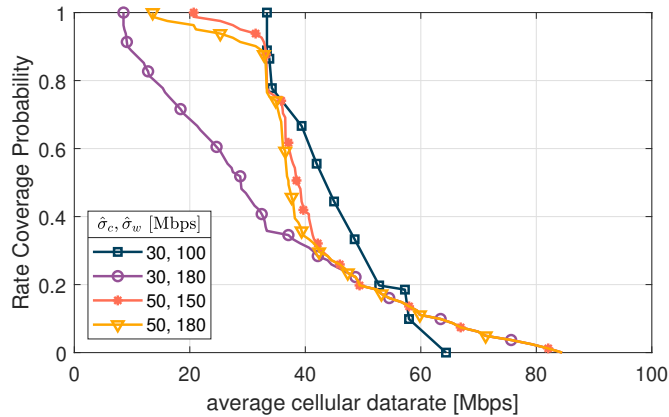}
    \caption{Empirical CCDF of average cellular datarate. Parameters: $|\mathcal{E}|=2$, ${\theta_c}/{\theta_w}= 7$.}
    \label{fig:cell_thresh}
\end{figure}

Next, we comment on the variation in the average WiFi datarate due to datarate thresholds. The rate coverage probability of average WiFi datarate is presented in Fig.~\ref{fig:wifi_thresh}, based on which we make the following observations:
\begin{itemize}
    \item The rate coverage probability is $1$ when the WiFi datarate threshold is low, i.e., $\hat{\sigma}_w = 100$~Mbps.
    \item For the same cellular datarate threshold $\hat{\sigma}_c$, the RCP increases on increasing the WiFi datarate threshold $\hat{\sigma}_w$. 
    \item The RCP decreases as the cellular datarate threshold $\hat{\sigma}_c$ is increased.
\end{itemize}

\begin{figure}[h!]
    \centering
    \includegraphics[width=0.6\columnwidth]{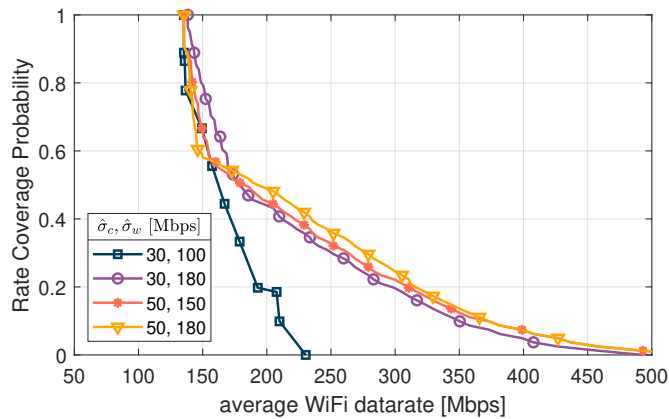}
    \caption{Empirical CCDF of average WiFi datarate. Parameters: $|\mathcal{E}|=2$, ${\theta_c}/{\theta_w}= 7$.}
    \label{fig:wifi_thresh}
\end{figure}

\subsection{Datarate Improvements}
{Here, we compare the performance of our algorithm \texttt{D-BRA} with a random strategy, where the values of $\delta_c$ and $\delta_w$ are randomly chosen from the closed set $[0.1, 1]$. We call this strategy RANDOM.
In Fig.~\ref{fig:boost_dat}, we clearly see that \texttt{D-BRA} outperforms RANDOM. Moreover, comparing the mean values of the average datarates, we find that \texttt{D-BRA} improves the performance of cellular networks by 11.37\% and of WiFi networks by 18.59\%.}

\begin{figure}[h!]
    \centering
    \includegraphics[width=0.6\columnwidth]{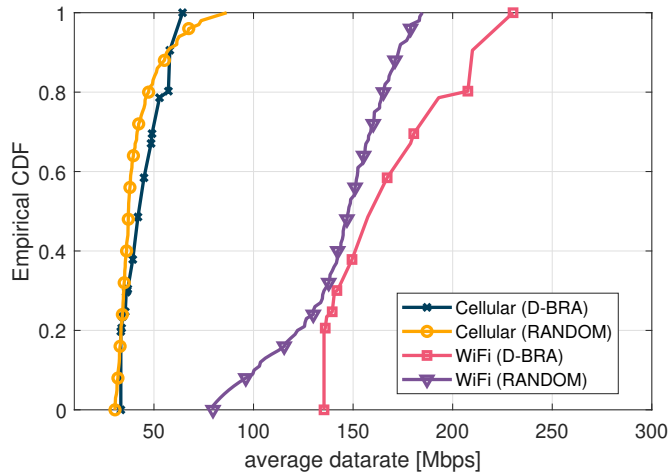}
    \caption{Empirical CDF of the average datarates. Parameters: ${\theta_c}/{\theta_w}\in \{5,6,7\}$, $|\mathcal{E}|=2$, $\hat{\sigma}_c = 30$~Mbps, $\hat{\sigma}_w = 100$~Mbps, $\gamma=10$~dB.}
    \label{fig:boost_dat}
\end{figure}

\subsection{Cellular vs. WiFi}
We consider a special case, where one entity owns the cellular network entirely and the other entity owns the complete WiFi network.
In Fig.~\ref{fig:cell_vs_wifi}, the datarate requirements of both entities are satisfied at equilibrium. The cellular network in both cases, settles at $\delta_c^1 = 1$, while the WiFi network settles at $\delta_w^2 \approx 0.4$ when $\hat{\sigma}_w^2 = 120$~Mbps in Fig.~\ref{cvw:a} and at $\delta_w^2 \approx 0.5$ when $\hat{\sigma}_w^2 = 130$~Mbps in Fig.~\ref{cvw:b}. Due to the increase in WiFi datarate threshold, the cellular datarate decreases in Fig.~\ref{cvw:b}, even though the minimum cellular datarate requirement is still met.

\begin{figure}[h!]%
 \centering
 \subfloat[$\hat{\sigma}_c^1=30$~Mbps, $\hat{\sigma}_w^2=120$~Mbps. ]{\includegraphics[width=0.8\columnwidth]{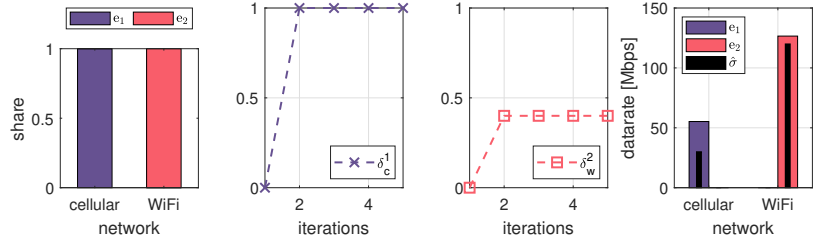}\label{cvw:a}}\\
 \subfloat[$\hat{\sigma}_c^1=45$~Mbps, $\hat{\sigma}_w^2=130$~Mbps. ]{\includegraphics[width=0.8\columnwidth]{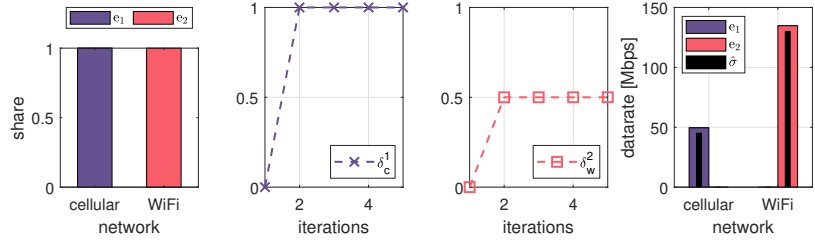}\label{cvw:b}}
 \caption{Interaction of two entities when one entity is entirely cellular and the other is entirely WiFi: $v_c^1=1$, $v_w^2=1$. Parameters: $|\mathcal{E}|=2$, $\theta_c = 7 \theta_w$.}%
 \label{fig:cell_vs_wifi}%
\end{figure}

\subsection{Three-Entity Game}
Next, we demonstrate that our framework works for more than two entities\footnote{{Currently, most countries have 3-4 major Telecom companies. For example, U.S. has AT\&T, Verizon and T-Mobile among the major players. Similarly, India has Reliance Jio, Airtel and Vodafone-Idea. Therefore, in our simulations, when we show the interaction of 2 entities and 3 entities, it very well captures the existing scenario.}}. Here three entities owning different shares of the cellular and WiFi networks interact such that the minimum datarate requirement of each entity is met at equilibrium.

\begin{figure}[h!]
    \centering
    \includegraphics[width=0.6\columnwidth]{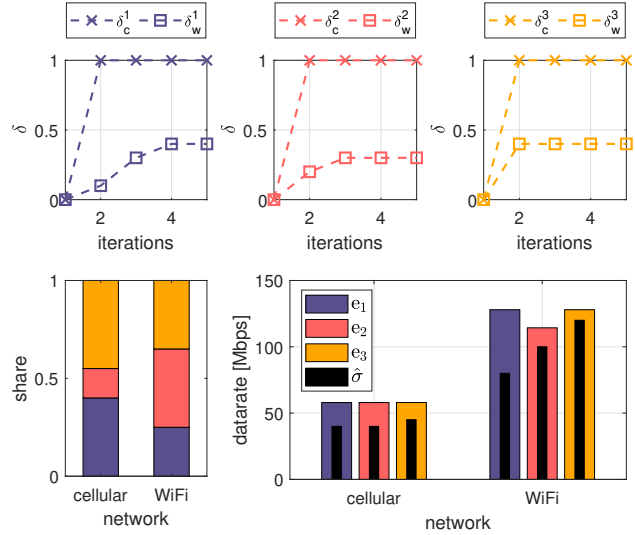}
    \caption{Exemplary three-entity interaction, $\theta_c = 7 \theta_w$.}
    \label{fig:3_ent}
\end{figure}

In Fig.~\ref{fig:3_ent}, the cellular and WiFi datarate thresholds are crossed by each entity. It must be noted that the entities have their WiFi datarate inversely proportional to the WiFi network share.

\newpage
\section{ {Case Study: Glasgow, UK}}
In this section we test our framework on real-world network of Glasgow, UK spanning $55.85^{\circ}$N to  $55.867^{\circ}$N in latitude and $4.265^{\circ}$W to $4.29^{\circ}$W in longitude. The choice of Glasgow is motivated by the availability of WiFi network data \cite{glasgowcitycouncil2014} published by the Glasgow City Council. The data for cell tower locations is downloaded from Open Cell ID \cite{opencellid2021}. However, to determine the positions of the incumbent users, we relied on manual data extraction based on the readings of FCC documents \cite{fcc2020new, fcc2020new2}.
We partition the set of all network elements (cellular and WiFi) among multiple entities based on the market shares\footnote{EE: 22\%, O2: 19\%, Vodafone + Three: 25\%, others: 34\%.} of different Telecommunication operators in UK~\cite[Table 36]{ofcom2021tt}. The case study captures the behavior of \texttt{D-BRA} when the spatial distribution of the users, access points and base stations are non uniform (see Fig.~\ref{fig:glasgow}).

\begin{figure}[h!]%
 \centering
 \subfloat[Locations of cellular BSs, WiFi APs and incumbent users.
 ]{\includegraphics[width=0.6\columnwidth]{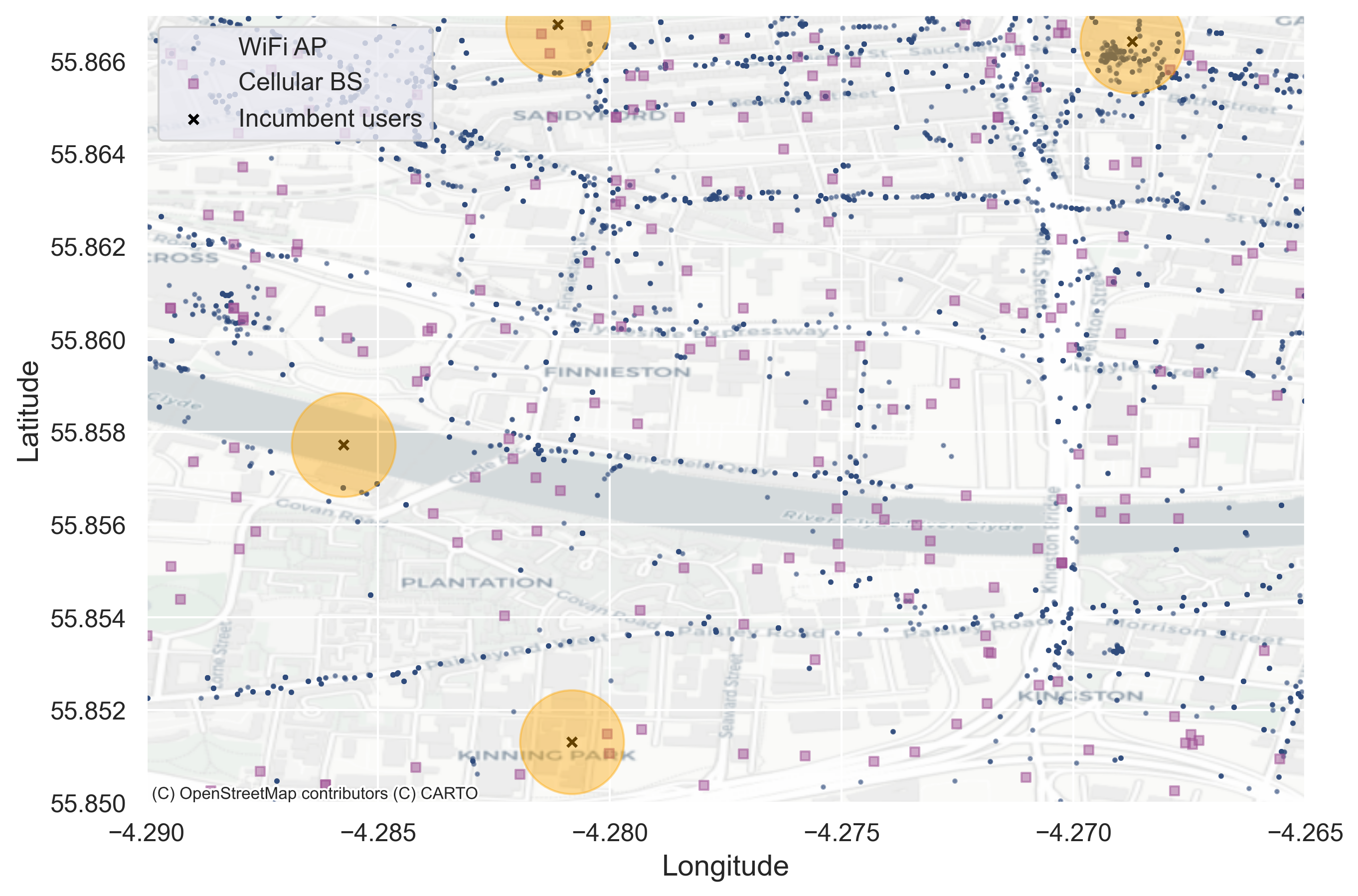}\label{fig:glasgowmap}} 
 \subfloat[Spatial distribution of population as per \cite{facebook2021}.
 ]{\includegraphics[width=0.4\columnwidth]{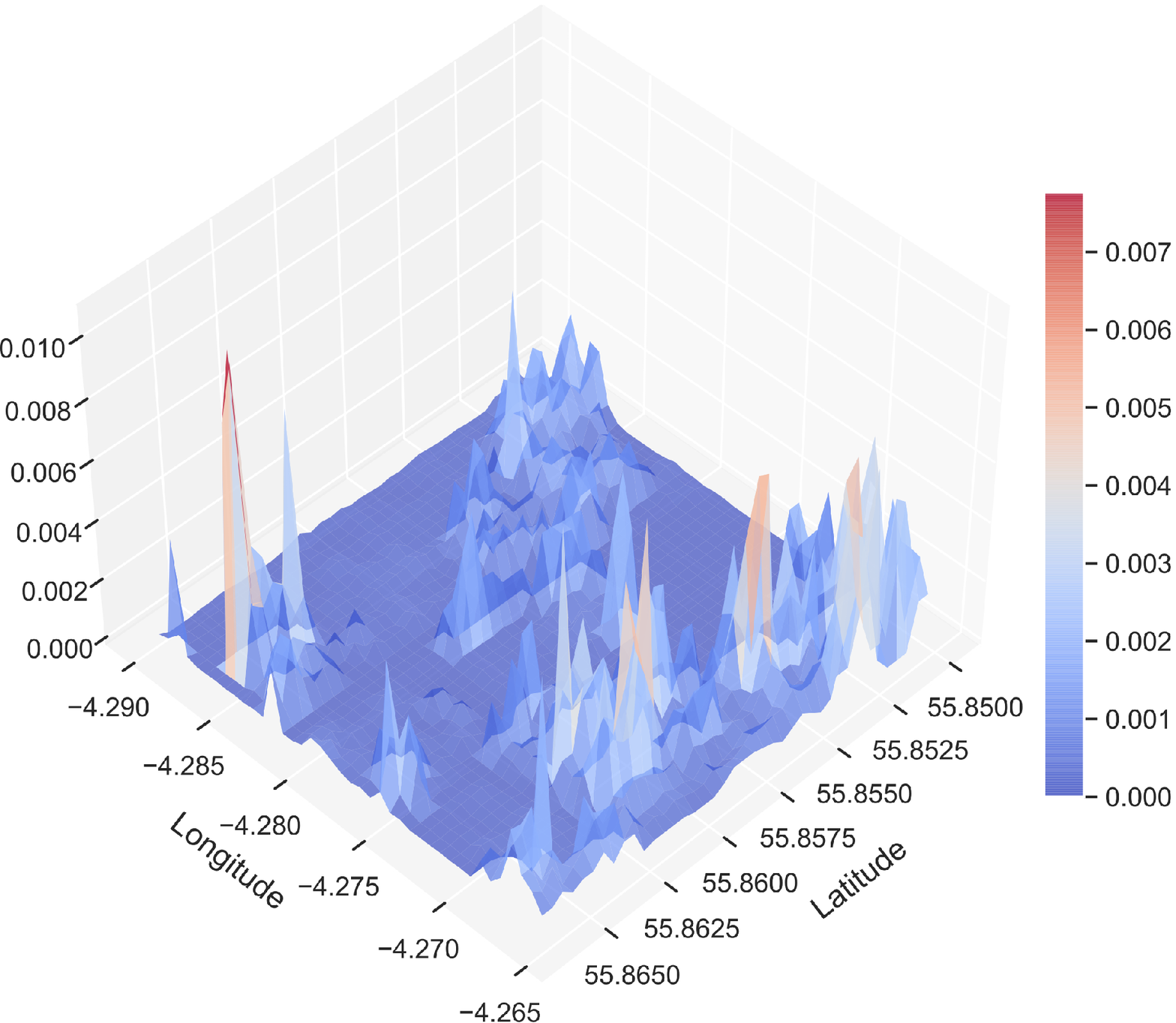}\label{fig:fbdata}}
 \caption{Real world spatial distribution of network elements and users over the map of Glasgow, UK.}%
 \label{fig:glasgow}%
\end{figure}

The network elements of each entity monitor the QoS delivered to its users through crowdsourcing. This way, the entities are able to evaluate the average QoS performance, and use that to compute the utility in the algorithm.

In Fig.~\ref{fig:glasgow_plots}, we show how the four entities adjust their action vectors to maximize their own utility. The action vectors do not stabilize, however the average datarate delivered by the cellular and WiFi networks of each entity become saturated. The results hint at the existence of a mixed strategy Nash equilibrium, which we can see if we let the game continue for a sufficiently large number of iterations. In Fig.~\ref{fig:4_ent}, we convey through box plots, the characteristics of the action vectors of each entity, and show a comparative bar plot of the average datarates achieved by them as a result of playing the game.

\begin{figure}[h!]%
 \centering
 \subfloat[Entity 1 ($v_c^1 = v_w^1 = 0.22$).
 ]{\includegraphics[width=0.5\columnwidth]{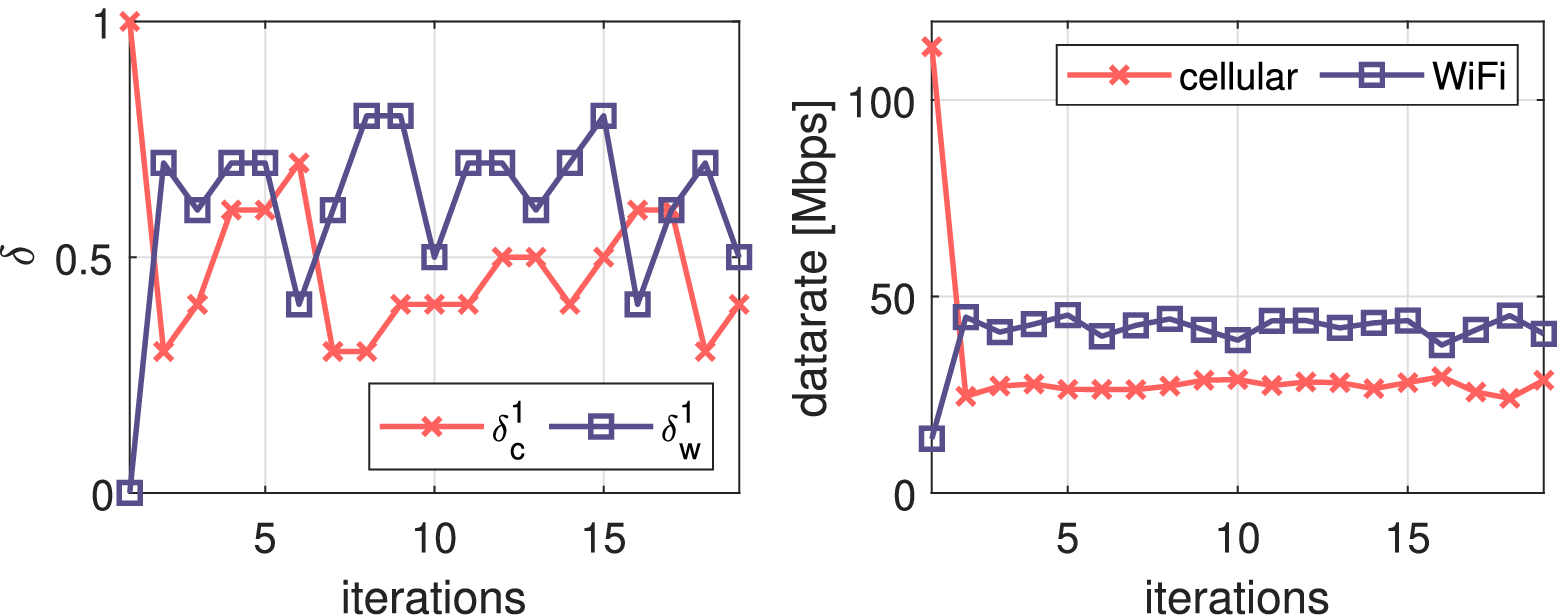}\label{cvw:a}} 
 \subfloat[Entity 2 ($v_c^2 = v_w^2 = 0.19$).
 ]{\includegraphics[width=0.5\columnwidth]{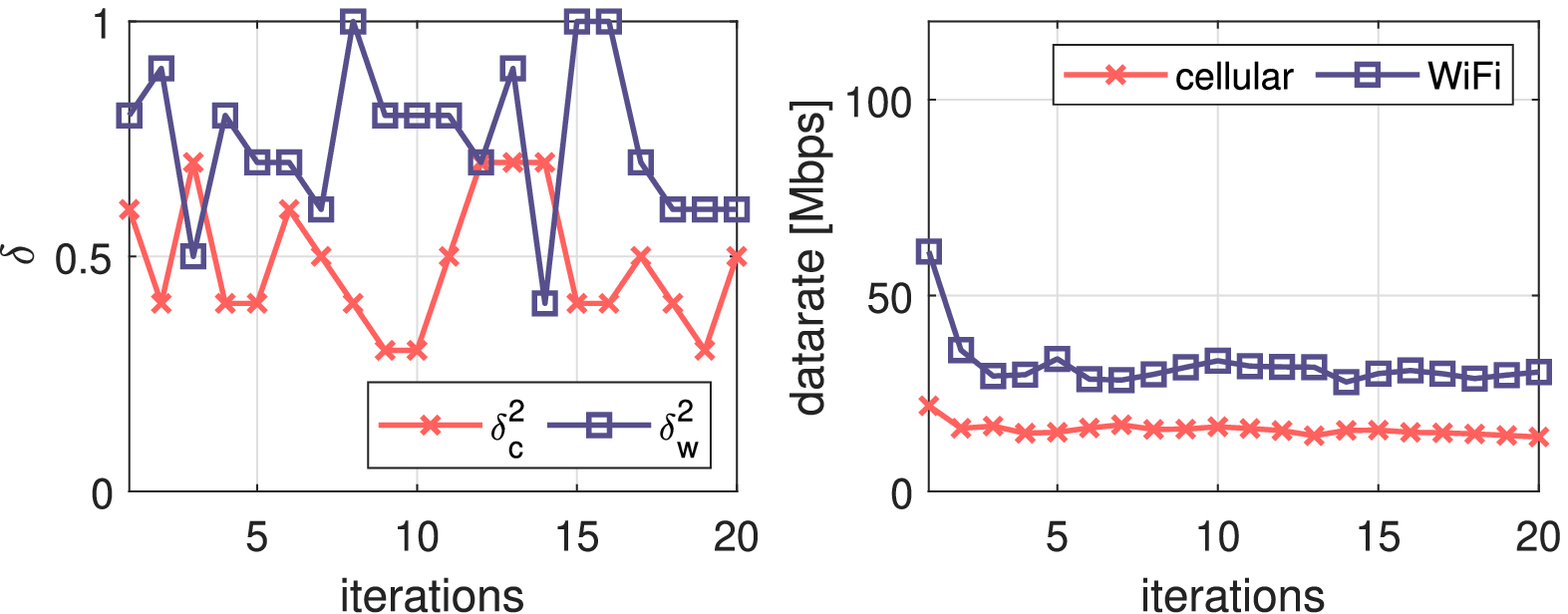}\label{cvw:b}}\\
 \subfloat[Entity 3 ($v_c^3 = v_w^3 = 0.25$).
 ]{\includegraphics[width=0.5\columnwidth]{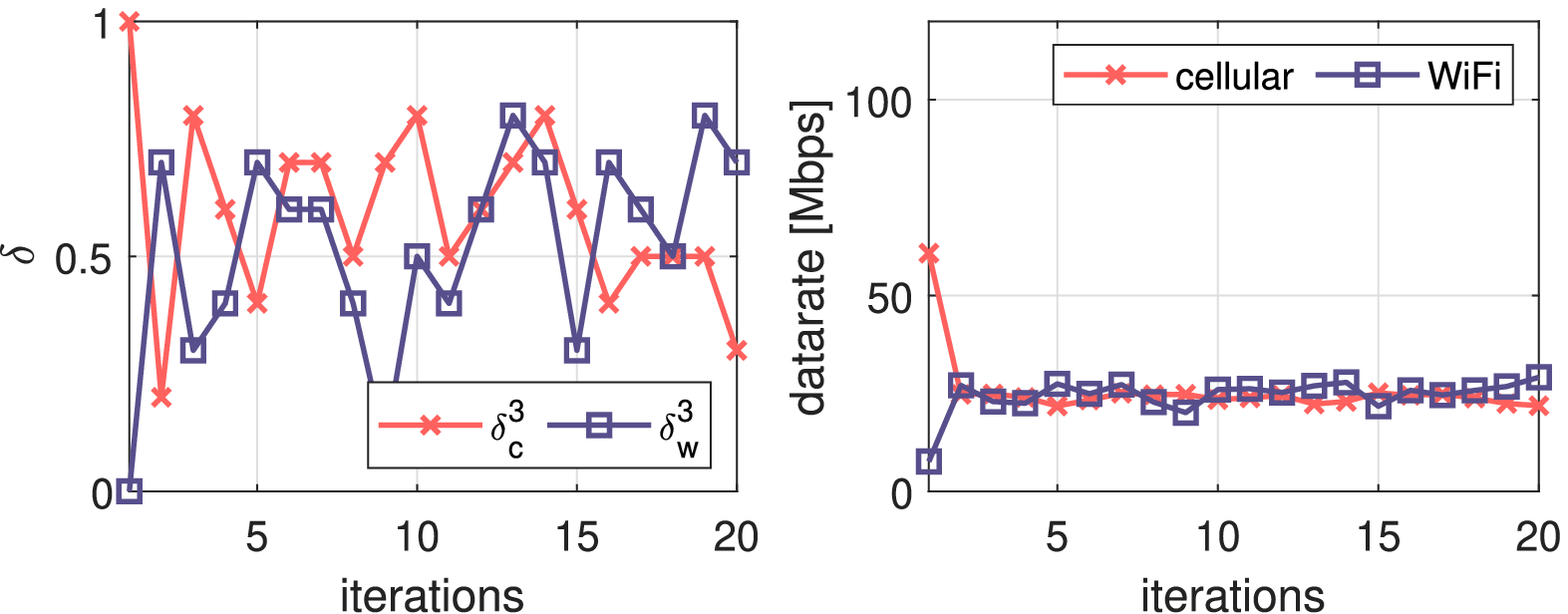}\label{cvw:c}} 
 \subfloat[Entity 4 ($v_c^4 = v_w^4 = 0.34$).
 ]{\includegraphics[width=0.5\columnwidth]{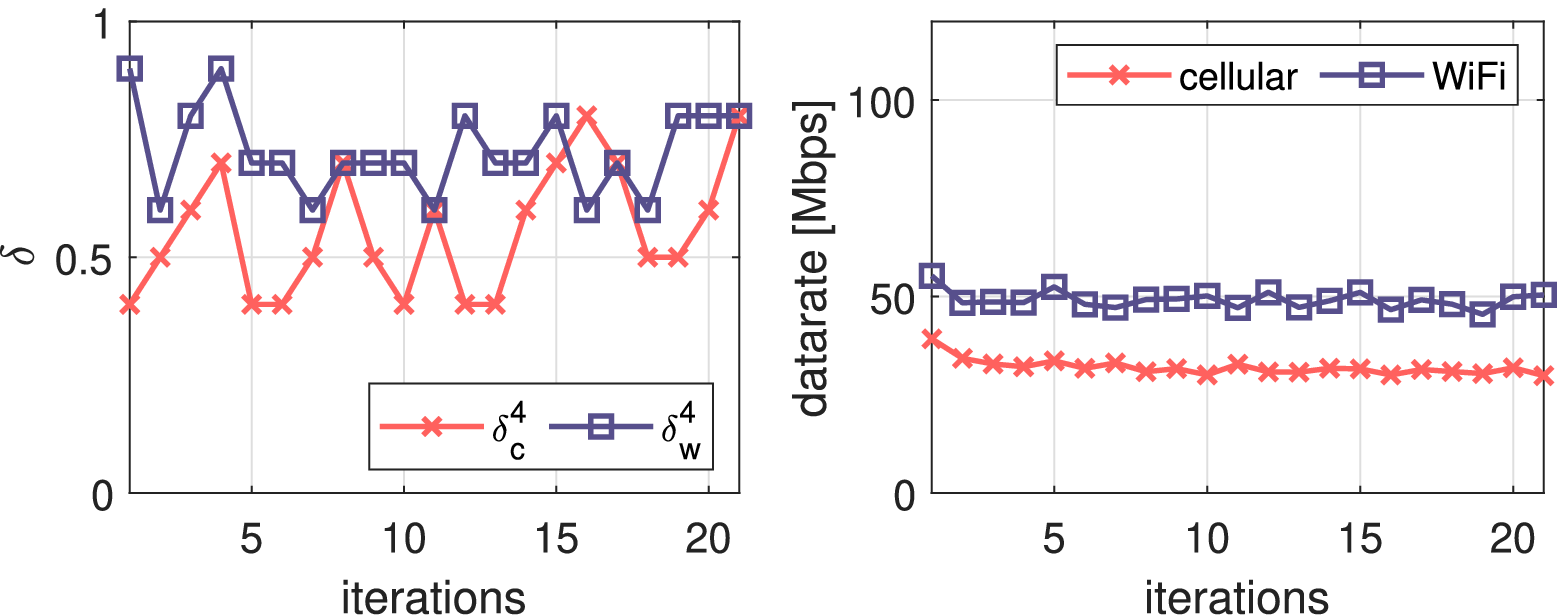}\label{cvw:d}}
 \caption{Interaction of four entities in the city of Glasgow. Parameters: $\theta_c = 5 \theta_w$.}%
 \label{fig:glasgow_plots}%
\end{figure}

\begin{figure}[h!]
    \centering
    \includegraphics[width=0.7\columnwidth]{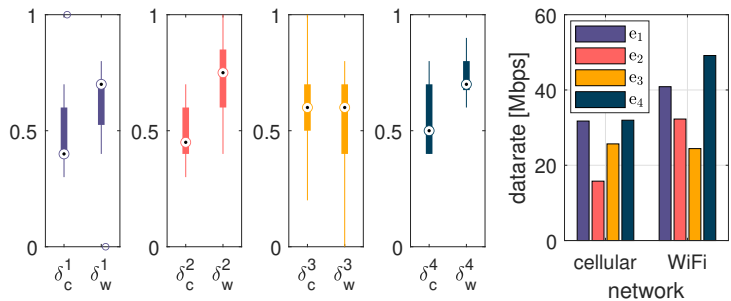}
    \caption{Box plot of action vector values of each entity and the resultant average datarate achieved by them.}
    \label{fig:4_ent}
\end{figure}

\section{Conclusion}
In this work, we provided a game-theoretic framework for analyzing the interaction between different networks, namely, cellular, WiFi or a combination of both, as they utilize the unlicensed spectrum in addition to their licensed spectra, in the presence of incumbent users. We have shown how the system parameters affect the performance of a network at equilibrium and quantified the {performance gains} of the networks as a result of using the 6~GHz bands, which offer larger bandwidths.
{The proposed distributed algorithm outperforms the random-strategic baseline and gives an average datarate boost of $\sim 11\%$ and $\sim 18\%$ for cellular and WiFi networks respectively.}
The proposed framework is flexible and can be used to model a variety of scenarios for feasibility and performance assessment of the networks involved. {For example, we can analyze how the deployment densities of the WiFi and cellular networks affect the equilibrium. In the absence of incumbent users, the framework reduces to multi-entity spectrum sharing. Moreover, the utility maximization problem can be solved using online learning algorithms~\cite{rahman2020online} instead of exhaustive search for a near optimal yet agile performance.}
{The theoretical analysis was done in a stochastic setup and we focused on a large-scale network. Lastly, we tested our framework on a real world location with real cellular network and WiFi access point deployments and showed that a practical implementation of multi-entity spectrum sharing is feasible.}



\appendix

\section{Laplace Transforms}\label{app:laplace}
 The Laplace transform of the interference \cite{haenggi2018stochastic} experienced by the cellular users in the unlicensed band from transmitters of type $j \in \{c,w,z\}$ is mathematically expressed as $\mathcal{L}_{c,j|U}(s)$:
\begingroup
\allowdisplaybreaks
\begin{align*}
    \mathcal{L}_{c,j|U}(s) &\stackrel{j \neq c}{=}  \exp \left( - \pi \lambda_{j|U} \dfrac{(s \, p_j)^{\frac{2}{\alpha}}}{ {\rm sinc}\left(\frac{2}{\alpha}\right) } \right), \\
    &\stackrel{j = c}{=}  \exp \left( - \int_{r}^{\infty} \dfrac{ 2\pi \lambda_{c|U} }{1 + (s\, p_c)^{-1}x^{\alpha}} x \, dx \right).
\end{align*}
\endgroup

Similarly, the Laplace transform of the interference experienced by the WiFi users in the unlicensed band from the transmitters of type $j \in \{c,w,z\}$ is mathematically expressed as $\mathcal{L}_{w,j|U}(s)$:
\begin{align*}
    \mathcal{L}_{w,j|U}(s) &=  \exp \left( - \pi \lambda_{j|U} \dfrac{(s \, p_j)^{\frac{2}{\alpha}}}{ {\rm sinc}\left(\frac{2}{\alpha}\right) } \right).
\end{align*}

\bibliographystyle{IEEEtran}
\bibliography{main.bib}

\end{document}